\documentclass[preprint,aps,superscriptaddress,showpacs,nofootinbib,preprintnumbers,amsmath,amssymb]{revtex4-1}
\usepackage{graphicx}% Include figure files
\usepackage{dcolumn}% Align table columns on decimal point
\usepackage{bm}% bold math
\usepackage{amsmath, amssymb}
\usepackage{commath}
\usepackage{MnSymbol}
\usepackage{amsthm}
\usepackage{hyperref}
\newtheorem*{theorem*}{Theorem}
\def\be{\begin{eqnarray}}
\def\ee{\end{eqnarray}}

\begin{document}

%\preprint{APS/123-QED}

\title{ The Buchdahl Stability Bound in Eddington-inspired Born-Infeld Gravity}

%author1
\author{Wei-Xiang Feng}
 \email{wfeng016@ucr.edu}
 \email{wxfeng@gapp.nthu.edu.tw}
 \affiliation{Department of Physics and Astronomy,
 University of California, Riverside, California 92521, USA}
%author2
\author{Chao-Qiang Geng}
 \email{geng@phys.nthu.edu.tw}
\affiliation{Department of Physics, National Tsing Hua University, Hsinchu 300, Taiwan}
\affiliation{Physics Division, National Center for Theoretical Sciences, Hsinchu 300, Taiwan}
\affiliation{Synergetic Innovation Center for Quantum Effects and Applications (SICQEA), 
Hunan Normal University, Changsha 410081, China}
%author3
\author{Ling-Wei Luo}
\email{lwluo@gate.sinica.edu.tw}
\affiliation{Department of Physics,
National Tsing Hua University, Hsinchu 300, Taiwan}
\affiliation{Institute of Physics, Academia Sinica, Taipei 11529, Taiwan}


\begin{abstract}
We give the Buchdahl stability bound in Eddington-inspired Born-Infeld (EiBI) gravity. 
We show that this bound depends on an energy condition controlled by the model parameter $\kappa$.
From this bound, we can constrain $\kappa\lesssim 10^{8}\text{m}^2$ if a neutron star with a mass around $3M_{\odot}$ is observed in the future.
In addition,  to avoid the potential pathologies in EiBI, a \emph{Hagedorn-like} equation of state associated with $\kappa$ at the center of a compact star is inevitable, which is similar to the Hagedorn temperature in string theory.
\end{abstract}

\maketitle

%\tableofcontents

\section{\label{sec:intro}introduction}
Despite general relativity (GR) being the most successful theory to describe gravity, 
there have been a number of  modified theories of GR to resolve various problems, such as 
 singularities, local energy problems and  incompatibilities with our quantum world. 
To regularize the singularities Deser and Gibbons (1998)~\cite{Deser:1998rj}, motivated by Born-Infeld electrodynamics~\cite{Born:1934gh}, discussed the possible forms of the gravitational analogue of Born-Infeld theory in the metric formalism, and the necessary conditions without running into the ghost problem. On the other hand, the ghost problem can also be avoided in the Palatini formalism.
In particular,
%motivated by Born-Infeld electrodyanmics~\cite{Born:1934gh}, 
Vollick (2004)~\cite{Vollick:2003qp} first considered a Born-Infeld-like action for gravity 
in the Palatini formalism. Later on, Ba{\~n}ados and Ferreira (2010)~\cite{Banados:2010ix} showed that this action can be regarded as the 
Eddington action if the matter action is absent~\cite{Eddington:1924}. 
Furthermore, when  matter \emph{solely} couples  to the metric, the highly nonlinear coupling between matter and gravity could avoid  the cosmological singularity in the radiation dominated era, 
resulting in significant deviations from GR in the extremely high density environment. 
Since then, this model has aroused much interest,  dubbed Eddington-inspired Born-Infeld (EiBI) gravity~\cite{BeltranJimenez:2017doy}.

It is known that
when we apply a modified theory of gravity to a stellar sphere, the junction conditions must be dealt with carefully. For example, in $f(R)$ gravity the matching of the first and second fundamental forms~\cite{Israel:1966rt} is not sufficient to merge the two spacetime regions
since the derivative terms of curvatures in the field equations play important roles~\cite{Deruelle:2007pt,Senovilla:2013vra}.
In contrast to the higher derivative theory, EiBI is equivalent to GR in vacuum, so that  the junction conditions are the same as
those in GR. This advantage makes EiBI an interesting modification of GR.
%%%%%%%
Nevertheless, when applying EiBI to compact stars, it has been found that 
there are some pathologies, such as the surface singularities~\cite{Pani:2012qd} , and  anomalies associated with phase transitions~\cite{Sham:2013sya}.
%, have been pointed out. 
However, the surface singularities can be prevented if the equations of state (EoS) is modified by the geodesic deviation (Jacobi) equation near the surface of a star~\cite{Kim:2013nna}. This aspect is quite similar to EoS
 in $f(R)$ theories, in which  more stringent junction conditions are needed~\cite{Ganguly:2013taa,Feng:2017hje}.

In GR, the Buchdahl stability bound implies that the radius of
 a stable stellar sphere cannot be smaller than $9/8$ of its gravitational radius~\cite{Buchdahl:1959zz}. 
 This inequality is independent of EoS contained in the sphere, which is  a generic feature of GR.
For EiBI gravity, it is not clear if there is a similar Buchdahl bound.
However, it is expected that the bound in GR should be different from that  in EiBI gravity.
%Therefore, it may be expected that the EoS contained inside a compact star will also modify the Buchdahl's bound in EiBI gravity, though, to the best of our knowledge the Buchdahl's bound in EiBI is still unclear. 
\section{\label{sec:EiBI}EiBI gravity}
To address this issue, we start from the EiBI action
~\cite{Banados:2010ix,Vollick:2003qp} with the geometric units $G=c=1$, given by
% is used throughout the Letter),
%\begin{align}
\be
\label{action}
S_{\text{EiBI}}[g, \Gamma]
&=&
\frac{2}{8\pi\kappa}\int d^{4}x\bigg[\sqrt{-\text{det}(\boldsymbol{g}+\kappa\boldsymbol{\mathcal{R}}
%\nonumber
(\Gamma))}-\lambda\sqrt{-\text{det}(\boldsymbol{g})}\bigg]+S_{M}[g,\Psi],
\ee
%nd{align}
where  $\boldsymbol{g}$ and $\boldsymbol{\mathcal{R}}$ correspond to the metric tensor and Ricci curvature based on the connection $\Gamma$
 with the matrix elements $g_{\mu\nu}$  and $\mathcal{R}_{\mu\nu}$, respectively, and
$S_{M}[g,\Psi] $ is the matter action, in which 
%constructed by coupling 
the generic matter field $\Psi$ couples only to the metric tensor $g$. 
Note that the cosmological constant is defined as $\Lambda=(\lambda-1)/\kappa$ in this model.
 However, when it comes to compact stars, it is reasonable to set $\lambda=1$ such that $\Lambda=0$. In the Palatini formalism, varying the action with respect to $g$ and $\Gamma$ independently gives the equations of motion of EiBI gravity as follows:
\begin{align}
q_{\mu\nu}&=g_{\mu\nu}+\kappa\mathcal{R}_{\mu\nu},\label{fieldeq1}\\
q^{\mu\nu}&=\tau(g^{\mu\nu}-8\pi\kappa T^{\mu\nu}),\label{fieldeq2}
\end{align}
where $\tau=\sqrt{\abs{\boldsymbol{g}}/\abs{\boldsymbol{q}}}$ with $\abs{\bullet}$ denoting the absolute value of the determinant of the matrix, 
$T^{\mu\nu}$ is the physical energy momentum tensor, and $q^{\mu\nu}\equiv q_{\mu\nu}^{-1}$.
Here, the auxiliary metric is used for raising/lowering the index in the geometric sector, while the metric tensor
is for the matter sector.
%, hence 
As a result, it follows that 
\begin{equation}\label{fieldeq3}
\begin{aligned}
\delta^{\mu}{}_{\nu}-\kappa\mathcal{R}^{\mu}{}_{\nu}=q^{\mu\lambda}g_{\lambda\nu}=\tau(\delta^{\mu}{}_{\nu}-8\pi\kappa T^{\mu}{}_{\nu}).
\end{aligned}
\end{equation}
By defining $\mathcal{R}\equiv \mathcal{R}^{\mu}{}_{\mu}$ and $T\equiv T^{\mu}{}_{\mu}$,
the field equations can be recast into the GR-like ones with the Einstein tensor built with the auxiliary metric~\cite{Delsate:2012ky},
given by
\begin{equation}\label{fieldeq4}
\begin{aligned}
\mathcal{G}^{\mu}{}_{\nu}[q]&=\mathcal{R}^{\mu}{}_{\nu}-\frac{1}{2}\delta^{\mu}{}_{\nu}\mathcal{R}\\
&=8\pi\big(\tau T^{\mu}{}_{\nu}+\mathcal{P}\delta^{\mu}{}_{\nu}\big)\equiv 8\pi\mathcal{T}^{\mu}{}_{\nu},
\end{aligned}
\end{equation}
where $\mathcal{P}=(\tau-1)/(8\pi\kappa)-\tau T/2$ is the isotropic pressure  and $\mathcal{T}^{\mu}{}_{\nu}$ is defined as the ``auxiliary'' energy momentum tensor.
Furthermore, by taking the determinant of Eq.(\ref{fieldeq3}), we can express  $\tau$ \emph{solely} in terms of $T^{\mu}{}_{\nu}$ as 
\begin{equation}\label{taueq1}
\tau=\left[\text{det}(\delta^{\mu}{}_{\nu}-8\pi\kappa T^{\mu}{}_{\nu})\right]^{-{1\over 2}}.
\end{equation}

%Perfect fluid
If we model the self-gravitating sphere by a perfect fluid with $T^{\mu}{}_{\nu}=(\rho+p)u^{\mu}u_{\nu}+p\delta^{\mu}{}_{\nu}$, where $\rho$, $p$ and $u^{\mu}$ denote the energy density, pressure and four-velocity of the fluid  with $u^{\mu}u_{\mu}=-1$, respectively, we can further 
write $\tau$ in terms of $\rho$ and $p$  as
\begin{equation}\label{taueq2}
\tau=\left[(1+8\pi\kappa\rho)(1-8\pi\kappa p)^{3}\right]^{-{1\over 2}}\equiv {1\over ab^3}\,,
\end{equation}
where $a\equiv\sqrt{1+8\pi\kappa\rho}$ and $b\equiv\sqrt{1-8\pi\kappa p}$ are required to be positive real number.
In addition, we can define the ``auxiliary'' density $\tilde{\rho}$ and pressure $\tilde{p}$ in terms of $a$ and $b$ by:
\begin{align}
\mathcal{T}^{t}{}_{t}&={-a^{2}+3b^{2}-2ab^{3}\over 16\pi\kappa ab^{3}}\equiv-\tilde{\rho},\label{appdensity1}\\
\mathcal{T}^{i}{}_{j}&={a^{2}+b^{2}-2ab^{3}\over 16\pi\kappa ab^{3}}\delta^{i}{}_{j}\equiv\tilde{p}\delta^{i}{}_{j}.\label{apppressure1}
\end{align}

%TOV Equation
\subsection{\label{sec:TOV}Tolman-Oppenheimer-Volkoff equation}
% Ansatz of the metric
In a spherically symmetric and static spacetime, the physical  and  auxiliary metrics are given by
\begin{equation*}
\begin{aligned}
g_{\mu\nu}dx^{\mu}dx^{\nu}&=-F^{2}(r)dt^{2}+G^{2}(r)dr^{2}+H^{2}(r)d\Omega^{2},\\
q_{\mu\nu}dx^{\mu}dx^{\nu}&=-A^{2}(r)dt^{2}+B^{2}(r)dr^{2}+r^{2}d\Omega^{2},
\end{aligned}
\end{equation*}
respectively, where $d\Omega^2=d\theta^2+\sin^2\theta~d\phi^2$.
For a perfect fluid, $T^{\mu\nu}=(\rho+p)u^{\mu}u^{\nu}+pg^{\mu\nu}$, the relations between the two sets of metrics $g_{\mu\nu}$ and $q_{\mu\nu}$ via Eq.(\ref{fieldeq2}) are given by
\begin{equation}
\begin{aligned}
F^{2}=A^{2}ab^{-3},\quad
G^{2}=B^{2}/ab,\quad
H^2=r^2/ab\,.
\end{aligned}
\end{equation}
Once  the auxiliary metric is solved, the physical one can be obtained immediately, and vice versa. 
We  note that these two metrics are identical to those  in the absence of matter, $i.e.$ vacuum, in which  $a=b=1$.

With the auxiliary metric at hand, we can also solve the auxiliary Einstein equations in Eq.(\ref{fieldeq4}) as in GR. 
First, we give an \emph{ans{a}tz} $B^2\equiv\big[1-2m(r)/r\big]^{-1}$, where $m(r)$ represents the auxiliary mass within the radial distance $r$. 
The $tt$ and  $rr-$components lead to
\begin{equation}\label{app.tt}
m'(r)=4\pi r^2\tilde{\rho}
\end{equation}
and
\begin{equation}\label{app.rr}
\frac{A'}{A}=\frac{m+4\pi r^{3}\tilde{p}}{r(r-2m)},
\end{equation}
respectively, where the $``'"$ denotes the derivative with respect to the coordinate $r$.
Together with $\tilde{\nabla}_{\alpha}\mathcal{T}^{\alpha}{}_{\beta}=0$ ($\tilde{\nabla}$ denotes the covariant derivative associated with $q_{\alpha\beta}$), we find
\begin{equation}\label{app.tov}
-\frac{1}{\tilde{\rho}+\tilde{p}}\frac{d\tilde{p}}{dr}=\frac{m+4\pi r^{3}\tilde{p}}{r(r-2m)},
\end{equation}
which is the Tolman-Oppenheimer-Volkoff (TOV) equation~\cite{Oppenheimer:1939ne} for auxiliary quantities, as expected. With the help of Eq.(\ref{apppressure1}), the derivative of $\tilde{p}(\rho,p)$ with respect to $r$ can be written in terms of $dp/dr$, 
resulting in the modified TOV equation
in EiBI gravity, given by
\begin{equation}\label{mtov}
-8\pi\kappa\bigg[\frac{(a^2-b^2)(b^2/c^{2}_{s}+3a^2)+4a^{2}b^{2}}{4a^{2}b^{2}(a^2-b^2)}\bigg]\frac{dp}{dr}=\frac{m+4\pi r^{3}\tilde{p}}{r(r-2m)},
\end{equation}
where $c_{s}^2\equiv dp/d\rho$ with $c_{s}$ the speed of sound.
From this equation, both $\kappa>0$ and $\kappa<0$ are allowed to support an equilibrium stellar structure~\cite{Pani:2011mg,Pani:2012qb,Sham:2012qi}.
 However, for $\kappa<0$, when some generic first-order phase transitions are present inside a star~\cite{Sham:2013sya},  the energy density $\rho$ can not be monotonically non-increasing throughout the star interior.
%This equation indicates that a negative value of $\kappa$ seems not possible to support an equilibrium stellar structure when generic phase transitions are present inside the star interior~\cite{Sham:2013sya}. 
As a result, we will only focus on the case of $\kappa>0$ in this study. 

By taking the derivative of $\tilde{p}$ in Eq.(\ref{app.rr}) and substituting the expression into Eq.(\ref{app.tov}), we find
\begin{equation}\label{Buch1}
\frac{d}{dr}\bigg[\frac{A'}{r}\sqrt{1-\frac{2m}{r}}\bigg]=\frac{A}{\sqrt{1-\frac{2m}{r}}}\bigg(\frac{m}{r^3}\bigg)',
\end{equation}
which characterizes how the average auxiliary density changes in accordance with the auxiliary potentials, and  plays the key role in proving the Buchdahl stability bound. Na{i}vely, we would expect the same Buchdahl's stability bound as in GR due to the GR-like field equations of the auxiliary quantities.
% except
% are now dealt with. 
However, in order to examine the  bound, we need to express these auxiliary quantities in terms of the physical ones.
Due to the highly nonlinear coupling of matter to gravity in EiBI, EoS
clearly affects the determination of the bound.

\subsection{\label{sec:eff}Effective density associated with the physical mass}
For a star submerged in the Schwarzshild vacuum, we use the Darmois-Israel matching conditions~\cite{Israel:1966rt} since EiBI is equivalent to
 GR in vacuum due to $B^{-2}(r)=A^{2}(r)=G^{-2}(r)=F^{2}(r)=1-2\mathcal{M}/r$, where $\mathcal{M}$ is the Schwarzschild mass in the spherically static spacetime. As the auxiliary mass $m(r)$ is quite different from the physical one with the mass function appearing in the physical metric inside the interior of a star, it should match $m(r_{s})=\mathcal{M}$ at radius $r_{s}$ with the condition $p(r_{s})=0$. 
 This suggests that we should take  another \emph{ans{a}tz} $G^2\equiv\big[1-2M(r)/r\big]^{-1}$ together with $G^2=B^2/ab$, where $M(r)$ stands for the physical mass, to obtain
\begin{equation}\label{mass relation}
M=m+\frac{1}{2}(1-ab)(r-2m),
\end{equation}
or equivalently
\begin{equation*}
m=M+\frac{1}{2ab}(ab-1)(r-2M).
\end{equation*}
The ``effective'' density associated with  $M(r)$ can be defined by $\rho_{\text{eff}}\equiv M'(r)/4\pi r^2$ in analogy to Eq.(\ref{app.tt}).
%, after some computations, 
Subsequently, we derive 
\begin{equation}\label{effdensity1}
\rho_{\text{eff}}=ab\tilde{\rho}+\frac{(1-ab)}{8\pi r^2}+\frac{\kappa}{2r}\bigg(1-\frac{2m}{r}\bigg)\bigg(\frac{a^2c_{s}^2-b^2}{ab}\bigg)\frac{d\rho}{dr}.
\end{equation}
Instead of the auxiliary density $\tilde{\rho}$, $\rho_{\text{eff}}$ is the one in EiBI,  corresponding to the physical density $\rho$ in GR. 
Therefore, the non-increasing monotonicity of $\rho_{\text{eff}}$ must be assumed to show the Buchdahl stability bound. 
On the other hand, it is worth noting that the second term of $(1-ab)/8\pi r^2$ in Eq.(\ref{effdensity1}) is potentially divergent as $r\rightarrow 0$ unless $ab=1$. 
To avoid the singularity of $\rho_{\text{eff}}$ in this model, we must have $ab=1$ at the center of a star as $r\rightarrow 0$. This situation is analogous to the surface singularities of the Ricci curvature associated with the metric tensor $g_{\mu\nu}$~\cite{Pani:2012qd,Kim:2013nna}.
%We will discuss this situation in due course.

\subsection{\label{sec:mono}Monotonicity of densities}
A crucial assumption in proving the Buchdahl stability bound in GR is the monotonically non-increasing property of the physical density $\rho$. However, 
%in EiBI the question is 
 %which one of $\rho$, $\tilde{\rho}$ and $\rho_{\text{eff}}$ 
  %should be monotonic in order to minimize our assumption. 
  %In other words, we have three types of density $\rho$, $\tilde{\rho}$ and $\rho_{\text{eff}}$ in this model. 
  is the monotonically non-increasing $\rho$ sufficient to have the same monotonic behavior of $\tilde{\rho}(\rho,p)$ and $\rho_{\text{eff}}(\rho,p)$?
To answer this question,
let us examine the differential relation between $\tilde{\rho}$ and $\rho$:
\begin{equation}\label{appdensity2}
\frac{d\tilde{\rho}}{dr}=\bigg[\frac{3a^2(a^2-b^2)c_{s}^2+(3b^2+a^2)b^2}{4a^3b^5}\bigg]\frac{d\rho}{dr}.
\end{equation}
The term in the numerator with $24\pi\kappa(\rho+p)a^2c_{s}^2+(3b^2+a^2)b^2>0$ is required to guarantee the non-increasing monotonicity of $\tilde{\rho}$ once $\rho$ is such. For a positive $\kappa$, the positivity is true
%there is no question from this condition 
if the null energy condition holds, $i.e.$ $\rho+p\geq0$.

For the effective density, taking the derivative of Eq.(\ref{effdensity1}) with respect to $r$ gives
\begin{equation}\label{effdensity2}
\begin{aligned}
\frac{d\rho_{\text{eff}}}{dr}=&ab\frac{d\tilde{\rho}}{dr}+\frac{(ab-1)}{4\pi r^3}+\kappa\bigg[\frac{1}{2r}\bigg(1-\frac{2m}{r}\bigg)\rho''\\
&+\bigg(\frac{2m}{r^3}-8\pi\tilde{\rho}\bigg)\rho'+\mathcal{O}(\kappa)\bigg]\bigg(\frac{a^2c_{s}^2-b^2}{ab}\bigg).
\end{aligned}
\end{equation}
Since the derivative  terms associated with $a$ and $b$ will pick out one order higher of $\kappa$, we collect them as $\mathcal{O}(\kappa)$ terms. Clearly, the first term in Eq.(\ref{effdensity2}) is negative if $\rho$ is monotonically non-increasing according to Eq.(\ref{appdensity2}). However, the assumption of the monotonically non-increasing $\rho$ cannot lead to the same monotonicity of $\rho_{\text{eff}}$ due to the potentially positive terms in Eq.(\ref{effdensity2}). 
To elaborate,
we define the $\kappa$ \emph{energy condition} $ab\geq 1$ ($\rho-p\geq 8\pi\kappa\rho p$) and consider if this condition is violated or not.
% (i)
For $ab\geq 1$ ($\rho-p\geq 8\pi\kappa\rho p$), $a^2c_{s}^2\leq b^2$, while in the GR case 
with $\kappa\rightarrow 0$, 
%resulting in 
$\rho-p\geq 0$ ($c_{s}^2\leq 1$), which is just the requirement for the  causal energy condition with positive density and pressure.
%; and (ii) 
On the other hand,  $ab<1$ ($\rho-p< 8\pi\kappa\rho p$) gives  $a^2c_{s}^2>b^2$, which violates the causality condition if $\kappa\rightarrow 0$, but is allowed for any finite value of $\kappa$. In these two cases, the net effect can be offset only if the terms in the square bracket of Eq.(\ref{effdensity2}) are positive, \emph{i.e.}
\begin{equation}
\frac{1}{2r}\bigg(1-\frac{2m}{r}\bigg)\rho''+\bigg(\frac{2m}{r^3}-8\pi\tilde{\rho}\bigg)\rho'>0.
\end{equation}

For $\mathcal{O}(\kappa)$, we are obliged to impose this condition as an additional assumption, if we want to interpret this effective density as the physical one corresponding to GR, \emph{i.e.} the density contributing to the \emph{physical} mass $M(r)$.
Since the profile of $\rho$ depends on EoS contained inside the star as well as the modified TOV equation in Eq.(\ref{mtov}), the assumption of 
the monotonically non-increasing $\rho_{\text{eff}}$ limits the possible classes of EoS in the EiBI theory.
%, though, further analysis is required elsewhere. In any case, these 
Note that the three densities tend to be the same in the GR limit of $\kappa\rightarrow 0$.
%, so it serves as a proper assumption as long as $\kappa\rho\ll 1$.

\section{\label{sec:Buch}The Buchdahl stability bound}
We now show the corresponding Buchdahl stability bound in the EiBI theory based on the assumptions made above.
%\begin{proposition}\label{prop}
To begin with, we can prove that if $\rho_{\text{eff}}$ is a monotonically non-increasing function, then
\begin{equation}\label{ineqM}
M(r)\geq\frac{r^3}{r_{s}^3}\mathcal{M},
\end{equation}
where $r_{s}$ is the radius of the star at which $p(r_{s})=0$ and $\mathcal{M}\equiv M(r_{s})$.
%\end{proposition}
\begin{proof}
By definition, $M(r)=\int^{r}_{0}4\pi\xi^2\rho_{\text{eff}}(\xi)d\xi$. Using the mean-value theorem, there exists $\bar{r}\in(0,r)$, such that
\begin{equation*}
\rho_{\text{eff}}(\bar{r})=\frac{\int^{r}_{0}4\pi\xi^2\rho_{\text{eff}}(\xi)d\xi}{\int^{r}_{0}4\pi\xi^2d\xi}\equiv \bar{\rho}_{\text{eff}}(r),
\end{equation*}
resulting in 
\begin{align}
M(r)&=\frac{4\pi r^3}{3}\bar{\rho}_{\text{eff}}(r)
%=\frac{r^3}{r_{s}^3}\bigg(\frac{4\pi r_{s}^3}{3}\bar{\rho}_{\text{eff}}(r)\bigg)\nonumber\\
%&
\geq\frac{r^3}{r_{s}^3}\bigg(\frac{4\pi r_{s}^3}{3}\bar{\rho}_{\text{eff}}(r_{s})\bigg)=\frac{r^3}{r_{s}^3}\mathcal{M},\nonumber
\end{align}
where 
%we have used 
$\bar{\rho}_{\text{eff}}(r)\geq\bar{\rho}_{\text{eff}}(r_{s})$ due to the non-increasing monotonicity, and $M(r_{s})=\mathcal{M}$.
\end{proof}
%\begin{corollary}\label{coroll1}
Together with Eq.(\ref{mass relation}), we readily obtain
\begin{equation}\label{Buch2}
m(r)+\frac{1}{2}(1-ab)[r-2m(r)]\geq\frac{r^3}{r_{s}^3}\mathcal{M}.
\end{equation}
%\end{corollary}
On the other hand,
%\begin{corollary}\label{coroll2}
if $\tilde{\rho}$ is a monotonically non-increasing function, then
\begin{equation}\label{Buch3}
M(r)+\frac{1}{2ab}(ab-1)[r-2M(r)]\geq\frac{r^3}{r_{s}^3}\mathcal{M},
\end{equation}
where we have used the fact that $m(r_{s})=M(r_{s})\equiv\mathcal{M}$ at the surface ($a=b=1$) of the star.
%\end{corollary}

From the inequalities in Eqs.~(\ref{Buch2}) and (\ref{Buch3}), we get the dual relations between the two mass functions. 
In order not to form a black hole, we must have $r-2m>0$ as well as $r-2M>0$ throughout the interior of the star. The signs of the extra terms in the inequalities,
which are absent in GR due to $a=b=1$, depend on the sign of $(ab-1)$.
If the $\kappa$ energy condition holds ($ab\geq1$), the inequality in Eq.~(\ref{Buch2}) is stronger than (or equal to) that in Eq.~(\ref{Buch3}), and the other way around if it is violated ($ab<1$).
%\emph{i.e.} the non-increasing monotonicity of $\rho_{\text{eff}}$ can imply Ineq.(\ref{ineqM}) but Ineq.~(\ref{Buch3}) cannot. In contrast, if $ab<1$, Ineq.~(\ref{Buch3}) is stronger than Ineq.~(\ref{Buch2}) since Ineq.~(\ref{Buch3}) gives a more stringent inequality than Eq.(\ref{ineqM}). 
When proving the Buchdahl stability bound, we will use the stricter assumption depending on the $\kappa$ energy condition.
% we are considering.

%\begin{lemma}
Furthermore, if $\tilde{\rho}$ is a monotonically non-increasing function, then
\begin{equation}\label{Buch4}
\bigg(\frac{m}{r^3}\bigg)'\leq 0.
\end{equation}
%\end{lemma}
\begin{proof}
\[
\bigg(\frac{m}{r^3}\bigg)'=\frac{m'}{r^3}-\frac{3m}{r^4}=\frac{4\pi}{r}\big[\tilde{\rho}(r)-\tilde{\rho}(\bar{r})\big]\leq 0,
\]
where $\bar{r}\in(0,r)$. Here, we have used the mean-value theorem  and  non-increasing monotonicity of $\tilde{\rho}$ to obtain the inequality.
\end{proof}
\begin{theorem*}
If (i) both $\tilde{\rho}$ and $\rho_{\text{eff}}$ are finite and monotonically non-increasing functions,
(ii) $A^2$ and $B^2$ are positive definite, and
(iii) the $\kappa$ energy condition ($ab\geq1$) holds,
the Buchdahl stability bound in EiBI gravity is given by
\begin{equation}
r_{s}\bigg(1-\frac{1}{2}g-\frac{1}{2}g^2\bigg)\geq\frac{9}{4}\mathcal{M},
\end{equation}
where 
\begin{equation}
\label{eqg}
g\equiv\frac{\mathcal{M}}{r_{s}^3}\int^{r_{s}}_{0}\frac{\sqrt{ab}-1}{\sqrt{1-\frac{2r^2}{r_{s}^3}\mathcal{M}}}rdr.
\end{equation}
\end{theorem*}
\begin{proof}
With  Eq.~(\ref{Buch4}) and $A(B)>0$, integrating Eq.(\ref{Buch1}) from $r$ to $r_{s}$, one gets
% gives
\[
\frac{A'(r_{s})}{r_{s}}\sqrt{1-\frac{2m(r_{s})}{r_{s}}}-\frac{A'(r)}{r}\sqrt{1-\frac{2m}{r}}\leq 0.
\]
To match the second fundamental form to the  Schwarzchild vacuum, it is required that
\[
A(r_{s})=\sqrt{1-\frac{2\mathcal{M}}{r_{s}}}\quad\text{and}\quad A'(r_{s})=\frac{\mathcal{M}}{r_{s}^2}\frac{1}{\sqrt{1-\frac{2\mathcal{M}}{r_{s}}}},
\]
%then we have
leading to
\[
A'(r)\geq\frac{r}{\sqrt{1-\frac{2m(r)}{r}}}\frac{\mathcal{M}}{r_{s}^3}\,,
\]
and
%integrating again to have
\[
A(r_{s})-A(0)\geq\frac{\mathcal{M}}{r_{s}^3}\int^{r_{s}}_{0}\frac{rdr}{\sqrt{1-\frac{2m(r)}{r}}}\equiv I.
\]
If the $\kappa$ energy condition ($ab\geq1$) holds, by writing  $m(r)$ in terms of $M(r)$ we have $r-2m=(r-2M)/ab$, which yields  $M(r)\geq\frac{r^3}{r_{s}^3}\mathcal{M}$ from
 Eq.~(\ref{Buch2}). As a result, we obtain
\begin{align}
I &=\frac{\mathcal{M}}{r_{s}^3}\int^{r_{s}}_{0}\frac{\sqrt{ab}rdr}{\sqrt{1-\frac{2M(r)}{r}}}\geq\frac{\mathcal{M}}{r_{s}^3}\int^{r_{s}}_{0}\frac{\sqrt{ab}rdr}{\sqrt{1-\frac{2r^2}{r_{s}^3}\mathcal{M}}}\nonumber\\
&=\frac{\mathcal{M}}{r_{s}^3}\int^{r_{s}}_{0}\frac{rdr}{\sqrt{1-\frac{2r^2}{r_{s}^3}\mathcal{M}}}+\frac{\mathcal{M}}{r_{s}^3}\int^{r_{s}}_{0}\frac{(\sqrt{ab}-1)rdr}{\sqrt{1-\frac{2r^2}{r_{s}^3}\mathcal{M}}}\nonumber\\
&=\frac{1}{2}\bigg[1-\sqrt{1-\frac{2\mathcal{M}}{r_{s}}}\bigg]+g.\nonumber
\end{align}
All of the above imply
\[
\sqrt{1-\frac{2\mathcal{M}}{r_{s}}}=A(r_{s})\geq A(r_{s})-A(0)\geq\frac{1}{2}\bigg[1-\sqrt{1-\frac{2\mathcal{M}}{r_{s}}}\bigg]+g,
\]
solving the inequality with the desired result.
\end{proof}
If the $\kappa$ energy condition ($ab\geq1$) holds throughout the interior of the sphere, $g$ is positive definite. This means that the lower bound of the stable radius in EiBI gravity is \emph{larger} than $(9/4)\mathcal{M}$. On the other hand, if the $\kappa$ energy condition ($ab<1$) is violated, we replace Eq.~(\ref{Buch3}) with $m(r)\geq\frac{r^3}{r_{s}^3}\mathcal{M}$ in the proof
to get the same bound as in GR. 

The intuitive explanation is the EoS \emph{switch-on} of the ``repulsive effect'' in EiBI gravity. It becomes clear by expanding the auxiliary quantities
in Eqs.(\ref{appdensity1}) and (\ref{apppressure1}) to the leading order of $\mathcal{O}(\kappa)$, given by
\begin{equation*}
\begin{aligned}
\tilde{\rho}&=\rho-\pi\kappa(5\rho^2-6\rho p-3p^2)+\mathcal{O}(\kappa^2),\\
\tilde{p}&=p+\pi\kappa(\rho^2+2\rho p+9p^2)+\mathcal{O}(\kappa^2).
\end{aligned}
\end{equation*}
We observe that the repulsive effect  with $\tilde{\rho}<\rho$ and $\tilde{p}>p$ in EiBI is significant only when $\rho>[(3+2\sqrt{3})/5]p\simeq 1.29p$. 
Note that we do not consider the case with
$\rho<[(3-2\sqrt{3})/5]p\simeq -0.09p$  if  the physical density and pressure are both positive inside the star.
%On the other hand,
 For $\tilde{\rho}>\rho$ and $\tilde{p}>p$, the effect is reduced due to the increased density  in $\tilde{\rho}$ compared to $\rho$. 
 When including all orders of $\kappa$, the two cases correspond to $ab>1$ and $ab<1$, respectively. 
% In this scenario, 
Note that the criteria for the repulsive effect to become significant depend on the $\kappa$ energy condition
%$i.e.$
and the specific value of $\kappa$, 
but the mechanism is the same as from
%we just look at 
the contributions of $\mathcal{O}(\kappa)$. 
In Fig. \ref{fig:I}, we plot the difference between the auxiliary and physical densities $8\pi \kappa (\tilde\rho-\rho) $ v.s. the corresponding  pressures
$8\pi \kappa (\tilde{p}-p)$,
 for $ab$ ranging from $0.85$ to $1.2$,
to show the regions where the repulsive effect is significant or not. The borderline near $ab\simeq 1$ corresponds to $\tilde{\rho}\simeq\rho$, where the repulsive effect switches on/off. Moreover, as we shall see, the critical value of $ab=1$  corresponds to an exotic EoS.
\begin{figure}
  %\centering
    \includegraphics[width=4in]{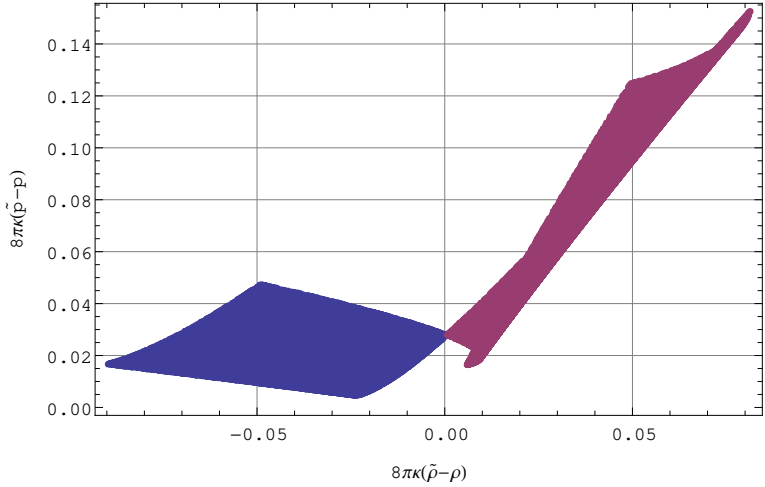}
  \caption{
  Difference between the auxiliary and physical densities $8\pi \kappa (\tilde\rho-\rho) $ v.s. the corresponding  pressures
$8\pi \kappa (\tilde{p}-p) $ in terms of  $ab$, where the blue (Left) and purple (Right) regions represent 
  $ab\gtrsim1$ with $1.1<a<1.2$ and $0.942<b<1.0$ and   $ab\lesssim1$ with $1.0<a<1.1$ and $0.85<b<0.942$, respectively.
 The borderline near $ab\simeq 1$ corresponds to $\tilde{\rho}\simeq\rho$. 
 % The repulsive effect in EiBI depends on whether 
  %%$\kappa$ energy condition holds 
  %$ab\geq1$ ($p\leq\rho/(1+8\pi\kappa\rho)$) or  $ab<1$ ($p>\rho/(1+8\pi\kappa\rho)$). 
  Here, $a\equiv\sqrt{1+8\pi\kappa\rho}$ and $b\equiv\sqrt{1-8\pi\kappa p}$, while 
  $\tilde{\rho}$ and $\tilde{p}$ are given in terms of $a$ and $b$ from Eqs.~(8) and (9), respectively.
  } 
 \label{fig:I}
\end{figure}

The ``singularity avoidance'' feature of this model~\cite{Delsate:2012ky} relies on the fact that as $b\rightarrow 0$, \emph{i.e.} $8\pi\kappa p\rightarrow 1$, the \emph{auxiliary} energy density (pressure) diverges but with a finite \emph{physical} energy density (pressure). In this regard, $\kappa$ can be taken as a cutoff near the Planck scale. However, as  noted previously, if we assume $b\neq 0$ inside the star,
% (since we expect things inside a compact star are still far from Planck scale even at the center of the star), 
$\rho_{\text{eff}}$ is still potentially divergent due to $(1-ab)/8\pi r^2$ as $r\rightarrow 0$. The remedy to cure the pathology is to set $ab=1$, or equivalently 
\begin{equation}
p=\frac{\rho}{1+8\pi\kappa\rho},
\end{equation}
around $r=0$. Physically, this leads to an exotic EoS controlled by $\kappa$ near the center ($r\lesssim\sqrt{\kappa}$) of a star \emph{regardless} of the real matter content. For this exotic EoS, the physical pressure $p$ is bounded by $1/8\pi\kappa$.
On the other hand,  there is no bound on the physical density $\rho$ , as $\rho\rightarrow\infty$ if 
$p\rightarrow 1/8\pi\kappa$. In other words, the cutoff of the physical pressure does not prevent the divergence of the physical density,
as shown in Fig.~\ref{fig:II}. Remarkably, this situation is in close analogy to the Hagedorn temperature~\cite{Hagedorn:1965st}, in which the energy and entropy diverge but with a fixed and finite (Hagedorn) temperature. This Hagedorn-like EoS~\cite{Gibbons:2001ck,Dubovsky:2012wk} manifests somewhat a deep connection of EiBI with the string theory~\cite{Dubovsky:2012wk,Atick:1988si,Brigante:2007jv}, which deserves further investigation. However, the discussions above are only at the classical level.
The pressure near the cutoff scale may signal a breakdown of EiBI or  Hagedorn-like phase transition. Whether the divergence of $\rho$ really occurs during the gravitational collapse~\cite{Harko:2003hs,Tavakoli:2015llr,Malafarina:2016yuf} requires a real  understanding of EiBI as well as its quantized version.

\begin{figure}
  %\centering
    \includegraphics[width=4in]{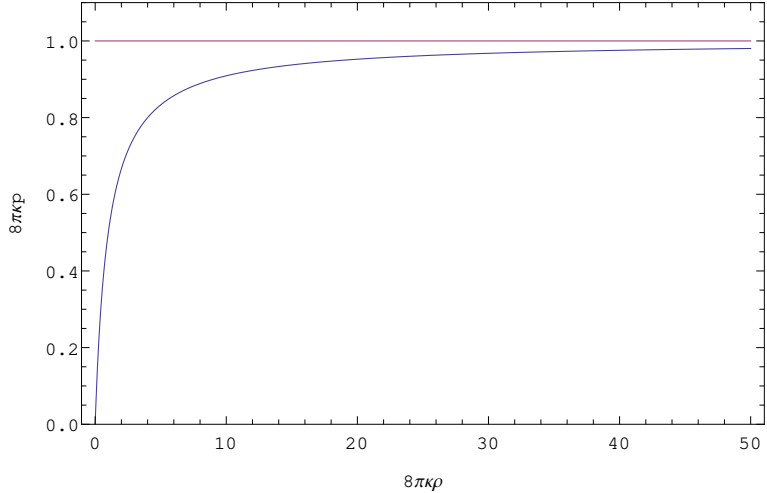}
  \caption{
 Exotic EoS behavior (blue curve) at the center of a star in the $\rho-p$ plane, which is analogous to the effect of the Hagedorn temperature,
where the  straight line (purple) represents the pressure bound $1/8\pi\kappa$
  } 
 \label{fig:II}
\end{figure}

\section{Conclusions}
The Buchdahl stability bound in EiBI is larger than (or equal to) $(9/4)\mathcal{M}$ in GR due to the \emph{repulsive effect} if 
the $\kappa$ energy condition ($ab\geq1$) holds throughout a star. To elaborate this statement, we see that $g$ in Eq.~(\ref{eqg})
can be calculated using the mean-value theorem such that
\begin{equation}
r_{s}\geq\bigg[2+\frac{\bar{a}\bar{b}/2}{1+\sqrt{\bar{a}\bar{b}}}\bigg]\mathcal{M},
\end{equation}
where $\bar{a}=\bar{a}(r_{s})\equiv a(\bar{r})$ and $\bar{b}=\bar{b}(r_{s})\equiv b(\bar{r})$, respectively, for $\bar{r}\in(0,r_{s})$. 
Expanding in terms of $\kappa$, we have $r_{s}/2\mathcal{M}\geq 9/8+(3\pi/8)\kappa(\bar{\rho}-\bar{p})+\mathcal{O}(\kappa^2)$. This inequality provides a direct way to constrain $\kappa$ just by examining the radii and masses of the most compact spherical objects in the sky. 
For instance, the typical density $\bar{\rho}\sim 10^{18}~\text{kg}/\text{m}^3$, radius $r_{s}\sim 12~\text{km}$ and gravitational radius $2\mathcal{M}\sim 6~\text{km}$ with $\mathcal{M}\sim 2M_{\odot}\sim 3~\text{km}$ of a neutron star (NS)~\cite{Lattimer:2004pg,Abbott:2018exr} yield the bound $\kappa\lesssim 10^{9}\text{m}^2$. Even though this gives a constraint on $\kappa$ of similar order of magnitude as those in Refs.~\cite{Avelino:2012ge,Pani:2011mg}, the bound on $\kappa$ can be further improved by \emph{one order} if a NS with the same radius but a mass around $3M_{\odot}$ is observed in the future.

On the other hand, if we believe that EiBI gravity is a viable theory down to the cutoff scale $\kappa$ , with the observed minimally stable radius of a spherical compact object of $(9/4)\mathcal{M}$, this would indicate that  an \emph{unusual} EoS, violating the $\kappa$ energy condition ($ab<1$), is contained inside the compact object. 
This sheds a new light on how we can probe EoS contained inside a compact star just by examining its smallest stable radius. However, 
the Buchdahl stability bound is modified, as that in GR, if an \emph{anisotropic} fluid is considered~\cite{Mak:2001eb,Ivanov:2017kyr,Raposo:2018rjn}. Such a situation merits further studies in the  future.

%\appendix
\section*{Appendix}

In Table~\ref{table:I}, we show the nomenclature and notation used in EiBI gravity.

\begin{table}
  \caption{Nomenclature and notation of EiBI gravity}
 %\centering
 \begin{ruledtabular}
 \begin{tabular}{l l}
  %\hline\hline
  % inserts table 
  %heading
  symbol  &  description   \\ [0.5ex] 
  \hline
$\kappa$ & EiBI model parameter (of length-squared dimension) \\
$g_{\mu\nu}$ & physical metric tensor \\ % inserting body of the table
$q_{\mu\nu}$ & auxiliary metric tensor \\
$\mathcal{R}_{\mu\nu}$ & Ricci tensor associated with $q_{\mu\nu}$ \\
$\mathcal{G}_{\mu\nu}$ & Einstein tensor associated with $q_{\mu\nu}$ \\
$T_{\mu\nu}$ & physical energy momentum tensor \\
$\mathcal{T}_{\mu\nu}$ & auxiliary energy momentum tensor \\
$\rho$ & physical energy density \\
$p$ & physical pressure \\
$a=\sqrt{1+8\pi\kappa\rho}$ & transformation factor between the two metrics depending on $\rho$\\
$b=\sqrt{1-8\pi\kappa p}$ & transformation factor between the two metrics depending on $p$ \\
$m$ & auxiliary mass function associated with $q_{\mu\nu}$ \\
$\tilde{\rho}$ & auxiliary energy density associated with $m$ \\
$\tilde{p}$ & auxiliary pressure \\
$M$ & physical mass function associated with $g_{\mu\nu}$ \\
$\rho_{\text{eff}}$ & effective energy density associated with $M$ \\
$r_s$ & radius of the self-gravitating sphere (star) \\
$\mathcal{M}=m(r_s)=M(r_s)$ & total mass of the self-gravitating sphere (star)\\ [1ex]
  %\hline\hline
 \end{tabular}
 \end{ruledtabular}
 \label{table:I}
\end{table}

\section*{Acknowledgements}
We thank Professor W.-T.~Ni for useful suggestions during a visit of W.-X.~Feng at National Center for Theoretical Sciences (NCTS), Taiwan. W.-X.~Feng was grateful to J.~C.~Baez and M.~Mulligan for helpful discussions as well as V.~Hubeny and C.-S.~Chu on the relevant issues during the Formosa Summer School on High Energy Physics 2018.
This work was supported in part by National Center for Theoretical Sciences and MoST (MoST-104-2112-M-007-003-MY3 and MoST-107-2119-M-007-013-MY3) and Academia Sinica Career Development Award Program (AS-CDA-105-M06).

%-----------------------------------------------------------------------%
% Bibliography
%-----------------------------------------------------------------------%


\begin{thebibliography}{99}

%\cite{Deser:1998rj}
\bibitem{Deser:1998rj} 
  S.~Deser and G.~W.~Gibbons,
  %``Born-Infeld-Einstein actions?,''
  Class.\ Quant.\ Grav.\  {\bf 15}, L35 (1998)
  %doi:10.1088/0264-9381/15/5/001
    
%\cite{Born:1934gh}
\bibitem{Born:1934gh} 
  M.~Born and L.~Infeld,
  %``Foundations of the new field theory,''
  Proc.\ Roy.\ Soc.\ Lond.\ A {\bf 144}, no. 852, 425 (1934)
  %doi:10.1098/rspa.1934.0059
  
%\cite{Vollick:2003qp}
\bibitem{Vollick:2003qp} 
  D.~N.~Vollick,
  %``Palatini approach to Born-Infeld-Einstein theory and a geometric description of electrodynamics,''
  Phys.\ Rev.\ D {\bf 69}, 064030 (2004)
  %doi:10.1103/PhysRevD.69.064030
  %[gr-qc/0309101].

%\cite{Banados:2010ix}
\bibitem{Banados:2010ix} 
  M.~Ba{\~n}ados and P.~G.~Ferreira,
  %``Eddington's theory of gravity and its progeny,''
  Phys.\ Rev.\ Lett.\  {\bf 105}, 011101 (2010)
  %Erratum: [Phys.\ Rev.\ Lett.\  {\bf 113}, no. 11, 119901 (2014)]
  %doi:10.1103/PhysRevLett.105.011101, 10.1103/PhysRevLett.113.119901
  %[arXiv:1006.1769 [astro-ph.CO]].

%\cite{Eddington:1924}
\bibitem{Eddington:1924}   
  A. S. Eddington, 
  The Mathematical Theory of Relativity (Cambridge University Press, Cambridge, England, 1924); 
  E. Schrödinger, Space-Time Structure (Cambridge University Press, Cambridge, 1985)
  
%\cite{BeltranJimenez:2017doy}
\bibitem{BeltranJimenez:2017doy} 
  J.~Beltran Jimenez, L.~Heisenberg, G.~J.~Olmo and D.~Rubiera-Garcia,
  %``Born?Infeld inspired modifications of gravity,''
  Phys.\ Rept.\  {\bf 727}, 1 (2018)
  %doi:10.1016/j.physrep.2017.11.001

  
 %\cite{Israel:1966rt}
\bibitem{Israel:1966rt} 
  W.~Israel,
  %``Singular hypersurfaces and thin shells in general relativity,''
  Nuovo Cim.\ B {\bf 44S10}, 1 (1966)
  %[Nuovo Cim.\ B {\bf 44}, 1 (1966)]
  %Erratum: [Nuovo Cim.\ B {\bf 48}, 463 (1967)].
  %doi:10.1007/BF02710419, 10.1007/BF02712210 
  
%\cite{Deruelle:2007pt}
\bibitem{Deruelle:2007pt} 
  N.~Deruelle, M.~Sasaki and Y.~Sendouda,
  %``Junction conditions in f(R) theories of gravity,''
  Prog.\ Theor.\ Phys.\  {\bf 119}, 237 (2008)
  %doi:10.1143/PTP.119.237
  %[arXiv:0711.1150 [gr-qc]].
  
%\cite{Senovilla:2013vra}
\bibitem{Senovilla:2013vra} 
  J.~M.~M.~Senovilla,
  %``Junction conditions for F(R)-gravity and their consequences,''
  Phys.\ Rev.\ D {\bf 88}, 064015 (2013)
  %doi:10.1103/PhysRevD.88.064015
  %[arXiv:1303.1408 [gr-qc]]. 
   
%\cite{Pani:2012qd}
\bibitem{Pani:2012qd} 
  P.~Pani and T.~P.~Sotiriou,
  %``Surface singularities in Eddington-inspired Born-Infeld gravity,''
  Phys.\ Rev.\ Lett.\  {\bf 109}, 251102 (2012)
  %doi:10.1103/PhysRevLett.109.251102
  %[arXiv:1209.2972 [gr-qc]].
  
%\cite{Sham:2013sya}
\bibitem{Sham:2013sya} 
  Y.~H.~Sham, P.~T.~Leung and L.~M.~Lin,
  %``Compact stars in Eddington-inspired Born-Infeld gravity: Anomalies associated with phase transitions,''
  Phys.\ Rev.\ D {\bf 87}, no. 6, 061503 (2013)
  %doi:10.1103/PhysRevD.87.061503
  %[arXiv:1304.0550 [gr-qc]].
  
%\cite{Kim:2013nna}
\bibitem{Kim:2013nna} 
  H.~C.~Kim,
  %``Physics at the surface of a star in Eddington-inspired Born-Infeld gravity,''
  Phys.\ Rev.\ D {\bf 89}, no. 6, 064001 (2014)
  %doi:10.1103/PhysRevD.89.064001
  %[arXiv:1312.0705 [gr-qc]].
  
%\cite{Ganguly:2013taa}
\bibitem{Ganguly:2013taa} 
  A.~Ganguly, R.~Gannouji, R.~Goswami and S.~Ray,
  %``Neutron stars in the Starobinsky model,''
  Phys.\ Rev.\ D {\bf 89}, no. 6, 064019 (2014)
  %doi:10.1103/PhysRevD.89.064019
  %[arXiv:1309.3279 [gr-qc]].
  
 %\cite{Feng:2017hje}
\bibitem{Feng:2017hje} 
  W.~X.~Feng, C.~Q.~Geng, W.~F.~Kao and L.~W.~Luo,
  %``Equation of State of Neutron Stars with Junction Conditions in the Starobinsky Model,''
  Int.\ J.\ Mod.\ Phys.\ D {\bf 27}, no. 01, 1750186 (2017)
  %doi:10.1142/S0218271817501863
  %[arXiv:1702.05936 [gr-qc]].
  
  %\cite{Buchdahl:1959zz}
\bibitem{Buchdahl:1959zz} 
  H.~A.~Buchdahl,
  %``General Relativistic Fluid Spheres,''
  Phys.\ Rev.\  {\bf 116}, 1027 (1959)
  %doi:10.1103/PhysRev.116.1027
  
  %\cite{Delsate:2012ky}
\bibitem{Delsate:2012ky} 
  T.~Delsate and J.~Steinhoff,
  %``New insights on the matter-gravity coupling paradigm,''
  Phys.\ Rev.\ Lett.\  {\bf 109}, 021101 (2012)
  %doi:10.1103/PhysRevLett.109.021101
  %[arXiv:1201.4989 [gr-qc]].
 
%\cite{Oppenheimer:1939ne}
\bibitem{Oppenheimer:1939ne} 
  J.~R.~Oppenheimer and G.~M.~Volkoff,
  %``On Massive neutron cores,''
  Phys.\ Rev.\  {\bf 55}, 374 (1939)
  %doi:10.1103/PhysRev.55.374
  
%\cite{Pani:2011mg}
\bibitem{Pani:2011mg} 
  P.~Pani, V.~Cardoso and T.~Delsate,
  %``Compact stars in Eddington inspired gravity,''
  Phys.\ Rev.\ Lett.\  {\bf 107}, 031101 (2011)
  %doi:10.1103/PhysRevLett.107.031101
  %[arXiv:1106.3569 [gr-qc]]. 

%\cite{Pani:2012qb}
\bibitem{Pani:2012qb} 
  P.~Pani, T.~Delsate and V.~Cardoso,
  %``Eddington-inspired Born-Infeld gravity. Phenomenology of non-linear gravity-matter coupling,''
  Phys.\ Rev.\ D {\bf 85}, 084020 (2012)
  %doi:10.1103/PhysRevD.85.084020
  %[arXiv:1201.2814 [gr-qc]].
  
%\cite{Sham:2012qi}
\bibitem{Sham:2012qi} 
  Y.-H.~Sham, L.-M.~Lin and P.~T.~Leung,
  %``Radial oscillations and stability of compact stars in Eddington inspired Born-Infeld gravity,''
  Phys.\ Rev.\ D {\bf 86}, 064015 (2012)
  %doi:10.1103/PhysRevD.86.064015
  %[arXiv:1208.1314 [gr-qc]].
  
%\cite{Hagedorn:1965st}
\bibitem{Hagedorn:1965st} 
  R.~Hagedorn,
  %``Statistical thermodynamics of strong interactions at high-energies,''
  Nuovo Cim.\ Suppl.\  {\bf 3}, 147 (1965)
 
%\cite{Gibbons:2001ck}
\bibitem{Gibbons:2001ck} 
  G.~W.~Gibbons,
  %``Pulse propagation in Born-Infeld theory: The World volume equivalence principle and the Hagedorn - like equation of state of the Chaplygin gas,''
  Grav.\ Cosmol.\  {\bf 8}, 2 (2002)
  
%\cite{Dubovsky:2012wk}
\bibitem{Dubovsky:2012wk} 
  S.~Dubovsky, R.~Flauger and V.~Gorbenko,
  %``Solving the Simplest Theory of Quantum Gravity,''
  JHEP {\bf 1209}, 133 (2012)
  %doi:10.1007/JHEP09(2012)133
  %[arXiv:1205.6805 [hep-th]].
    
 %\cite{Atick:1988si}
\bibitem{Atick:1988si} 
  J.~J.~Atick and E.~Witten,
  %``The Hagedorn Transition and the Number of Degrees of Freedom of String Theory,''
  Nucl.\ Phys.\ B {\bf 310}, 291 (1988)
  %doi:10.1016/0550-3213(88)90151-4
  
%\cite{Brigante:2007jv}
\bibitem{Brigante:2007jv} 
  M.~Brigante, G.~Festuccia and H.~Liu,
  %``Hagedorn divergences and tachyon potential,''
  JHEP {\bf 0706}, 008 (2007)
  %doi:10.1088/1126-6708/2007/06/008
  %[hep-th/0701205].
  
%\cite{Harko:2003hs}
\bibitem{Harko:2003hs} 
  T.~Harko,
  %``Gravitational collapse of a hagedorn fluid in vaidya geometry,''
  Phys.\ Rev.\ D {\bf 68}, 064005 (2003)
  %doi:10.1103/PhysRevD.68.064005
  %[gr-qc/0307064].
  
%\cite{Tavakoli:2015llr}
\bibitem{Tavakoli:2015llr} 
  Y.~Tavakoli, C.~Escamilla-Rivera and J.~C.~Fabris,
  %``The final state of gravitational collapse in Eddington?inspired Born?Infeld theory,''
  Annalen Phys.\  {\bf 529}, no. 5, 1600415 (2017)
  %doi:10.1002/andp.201600415
  %[arXiv:1512.05162 [gr-qc]].

%\cite{Malafarina:2016yuf}
\bibitem{Malafarina:2016yuf} 
  D.~Malafarina,
  %``Gravitational collapse of Hagedorn fluids,''
  Phys.\ Rev.\ D {\bf 93}, no. 10, 104042 (2016)
  %doi:10.1103/PhysRevD.93.104042
  %[arXiv:1605.03312 [gr-qc]].

%\cite{Lattimer:2004pg}
\bibitem{Lattimer:2004pg} 
  J.~M.~Lattimer and M.~Prakash,
  %``The physics of neutron stars,''
  Science {\bf 304}, 536 (2004)
  %doi:10.1126/science.1090720
  %[astro-ph/0405262].
    
%\cite{Abbott:2018exr}
\bibitem{Abbott:2018exr} 
  B.~P.~Abbott {\it et al.} [LIGO Scientific and Virgo Collaborations],
  %``GW170817: Measurements of neutron star radii and equation of state,''
  Phys.\ Rev.\ Lett.\  {\bf 121}, no. 16, 161101 (2018)
  %doi:10.1103/PhysRevLett.121.161101
  %[arXiv:1805.11581 [gr-qc]].
  
%\cite{Avelino:2012ge}
\bibitem{Avelino:2012ge} 
  P.~P.~Avelino,
  %``Eddington-inspired Born-Infeld gravity: astrophysical and cosmological constraints,''
  Phys.\ Rev.\ D {\bf 85}, 104053 (2012)
  %doi:10.1103/PhysRevD.85.104053
  %[arXiv:1201.2544 [astro-ph.CO]].
  
%\cite{Mak:2001eb}
\bibitem{Mak:2001eb} 
  M.~K.~Mak and T.~Harko,
  %``Anisotropic stars in general relativity,''
  Proc.\ Roy.\ Soc.\ Lond.\ A {\bf 459}, 393 (2003)
  %doi:10.1098/rspa.2002.1014
  %[gr-qc/0110103].

%\cite{Ivanov:2017kyr}
\bibitem{Ivanov:2017kyr} 
  B.~V.~Ivanov,
  %``Analytical study of anisotropic compact star models,''
  Eur.\ Phys.\ J.\ C {\bf 77}, no. 11, 738 (2017)
  %doi:10.1140/epjc/s10052-017-5322-7
  %[arXiv:1708.07971 [gr-qc]].
  
%\cite{Raposo:2018rjn}
\bibitem{Raposo:2018rjn} 
  G.~Raposo, P.~Pani, M.~Bezares, C.~Palenzuela and V.~Cardoso,
  %``Anisotropic stars as ultracompact objects in General Relativity,''
  Phys.\ Rev.\ D {\bf 99}, no. 10, 104072 (2019)
  %doi:10.1103/PhysRevD.99.104072
  %[arXiv:1811.07917 [gr-qc]].
  
\end{thebibliography}
\end{document}